\newif\ifnotes
\notestrue  

\documentclass[11pt]{article}
\usepackage[T1]{fontenc}
\usepackage[font=small]{caption} 
\usepackage{fullpage}
\usepackage[noend]{algorithmic}
\usepackage{algorithm}
\usepackage{amsfonts}
\usepackage{amsmath,amsthm}
\usepackage{amssymb}
 \usepackage{array,multirow,graphicx}
\usepackage{enumitem}
\usepackage[normalem]{ulem}
\usepackage{xcolor}
\usepackage{mdframed}
\usepackage{framed}
\usepackage{xspace}
\usepackage{marginnote}
\usepackage{soul}
\usepackage{breakcites}
\usepackage{mleftright}

\usepackage[pdfstartview=FitH,pdfpagemode=None,colorlinks,linkcolor=blue,filecolor=blue,citecolor=gray,urlcolor=red]{hyperref}
\usepackage[capitalize]{cleveref}
\usepackage{multicol}
\usepackage{geometry}
\usepackage{palatino} 
\fontfamily{ppl}\selectfont

\usepackage[textsize=tiny]{todonotes}
\newcommand{\mnote}[1]{\ifnotes{\color{purple}$\ll$\textbf{Mor}: {#1}$\gg$}\fi}

\newcommand{\nnote}[1]{\ifnotes{\color{blue}$\ll$\textbf{Noga}: {#1}$\gg$}\fi}

\newif\iffull
\fulltrue  
\newif\ifconf
\conffalse 

\newcommand{\poly}{{\mathsf{poly}}}

\newcommand{\img}{{\mathrm{image}}}

\newcommand{\test}{\mathrm{TEST}}

\theoremstyle{plain}
\newtheorem{thm}{Theorem}[section]      
\newcommand{\BT}{\begin{thm}}   \newcommand{\ET}{\end{thm}}
\newtheorem{dfn}[thm]{Definition}      %
\newcommand{\BD}{\begin{dfn}}   \newcommand{\ED}{\end{dfn}}
\newtheorem{corr}[thm]{Corollary}      %
\newcommand{\BCR}{\begin{corr}} \newcommand{\ECR}{\end{corr}}
\newtheorem{conj}[thm]{Conjecture}      %
\newcommand{\BCO}{\begin{conj}} \newcommand{\ECO}{\end{conj}}
\newtheorem{construction}[thm]{Construction}      %
\newcommand{\BCON}{\begin{construction}} \newcommand{\ECON}{\end{construction}}

\newtheorem{Ithm}{Theorem}
\newcommand{\BIT}{\begin{Ithm}}   \newcommand{\EIT}{\end{Ithm}}
\newtheorem{lem}[thm]{Lemma}
\newcommand{\BL}{\begin{lem}}   \newcommand{\EL}{\end{lem}}
\newtheorem{prop}[thm]{Proposition}
\newcommand{\BP}{\begin{prop}}   \newcommand{\EP}{\end{prop}}
\newtheorem{clm}[thm]{Claim}            %
\newcommand{\BCM}{\begin{clm}}   \newcommand{\ECM}{\end{clm}}
\newtheorem{fact}[thm]{Fact}            %
\newcommand{\BF}{\begin{fact}}   \newcommand{\EF}{\end{fact}}
\renewenvironment{proof}{\noindent{\bf Proof:~~}}{\qed}
\newcommand{\BPF}{\begin{proof}} \newcommand {\EPF}{\end{proof}}
\newtheorem{prot}{Protocol}      
\newcommand{\BPR}{\begin{prot}}   \newcommand{\EPR}{\end{prot}}
\newenvironment{cproof}{\noindent{\bf Proof:~~}}{\hfill $\Box$}
\newcommand{\BCPF}{\begin{cproof}} \newcommand {\ECPF}{\end{cproof}}

\newtheorem{remarkk}[thm]{Remark}            %
\newcommand{\BR}{\begin{remarkk}}   \newcommand{\ER}{\end{remarkk}}
\newcommand{\BDE}{\begin{description}}
\newcommand{\EDE}{\end{description}}
\newcommand{\BE}{\begin{enumerate}}
\newcommand{\EE}{\end{enumerate}}
\newcommand{\BI}{\begin{itemize}}
\newcommand{\EI}{\end{itemize}}
\newcommand{\BEQ}{\begin{eqnarray*}}
\newcommand{\EEQ}{\end{eqnarray*}}

\def\blackslug
{\hbox{\hskip 1pt\vrule width 8pt height 8pt depth 1.5pt\hskip 1pt}}
\def\qed{\quad\blackslug\lower 8.5pt\null\par}

\newenvironment{boxfig}[2]{\begin{figure}[#1]\fbox{\begin{minipage}{.95\columnwidth}
                        \vspace{0.2em}
                        \makebox[0.025\columnwidth]{}
                        \begin{minipage}{0.95\columnwidth}
            {\small{
                        #2 }}
                        \end{minipage}
                        \vspace{0.2em}
                        \end{minipage}}}{\end{figure}}

\newcommand{\cI}{{\cal I}}

\newcommand{\textdef}[1]{\textnormal{\textsf{#1}}}

\newcommand{\FF}{\mathbb{F}}

\newcommand{\NN}{\mathbb{N}}

\newcommand{\size}[1]{\left|{#1}\right|}

\newcommand{\Set}[1]{\left[#1\right]}

\newcommand{\Enc}{{\mathsf{Enc}}}

\newcommand{\msg}{{\sf m}}
\newcommand{\Supp}{{\sf Supp}}

\newcommand{\wt}{\mathrm{wt}}

\title{Zero-knowledge and Quantum CSS Codes Are Equivalent}

\author{ Noga Ron-Zewi\thanks{Department of Computer Science, University of Haifa. Email: \texttt{noga@cs.haifa.ac.il}. } \and Mor Weiss \thanks{Faculty of Engineering, Bar-Ilan University. Email: \texttt{mor.weiss@biu.ac.il}.}}

\date{}

\begin{document}

\maketitle
 \thispagestyle{empty}
\pagestyle{plain}


\begin{abstract}
Zero-knowledge codes, introduced by Decatur, Goldreich, and Ron~\cite{DecaturGR20}, are error-correcting codes in which few codeword symbols reveal no information about the encoded message, and have been extensively used in cryptographic constructions. Quantum CSS codes, introduced by Calderbank and Shor~\cite{CS96} and Steane~\cite{Ste96}, are error-correcting codes that allow for quantum error correction, and are also useful for applications in quantum complexity theory.
In this short note, we show that (linear, perfect) zero-knowledge codes and quantum CSS codes are equivalent. We demonstrate the potential of this equivalence by using it to obtain explicit asymptotically-good zero-knowledge locally-testable codes.
\end{abstract}

\section{Introduction}

In this note, we show an equivalence between two well-studied families of codes: Zero-knowledge codes and quantum CSS codes. 
We first briefly describe these families of codes and their applications.

\paragraph{Zero-knowledge codes.}
Zero-Knowledge (ZK) codes are error-correcting codes with a randomized encoding, in which few codeword symbols reveal nothing about the encoded message. More accurately, a 
$t$-\emph{Zero-Knowledge} (ZK) code $C\subseteq\FF^n$, for some $t\in\NN$, is associated with a randomized encoding map $\Enc_C:\FF^k\rightarrow\FF^n$ and has the following guarantee. For every $\msg,\msg'\in\FF^k$, and any $\cI\subseteq\Set{n}$ of size $\size{\cI}\leq t$, we have that $\Enc_C(\msg)|_\cI$ and $\Enc_C(\msg')|_\cI$ are identically distributed, where $\Enc_C(\msg)|_\cI$ denotes the restriction to $\cI$ of a randomly-generated encoding of $m$.  In this note, we focus on {\em non-adaptive, perfect} ZK codes, as defined above. That is, the queries of an adversary to the codeword are determined non-adaptively, and the resultant distribution is identical for every pair of messages. We note that the non-adaptive and adaptive settings are known to be equivalent~\cite[Appendix C]{BCL20}, and that relaxations to the statistical setting (where the distributions $\Enc_C(\msg)|_\cI, \Enc_C(\msg')|_\cI$ are statistically close) have also been considered in the literature~\cite{IshaiSVW13}.  We additionally focus on \emph{linear} ZK codes, meaning that $C$ is a linear subspace of $\FF^n$, and the encoding map $\Enc_C$ is linear in the message $\msg$ and the randomness used for encoding.

ZK codes were first formally defined by Decatur, Goldreich, and Ron~\cite{DecaturGR20}, and have been used (either explicitly or implicitly) in numerous applications in cryptography. For example, these codes have been used  in Shamir's secret sharing~\cite{Shamir79}, and  lie at the heart of information-theoretically secure Multi-Party Computation (MPC) protocols (starting from~\cite{Ben-OrGW88,ChaumCD88}). ZK codes also have applications to memory delegation (e.g., in PIR schemes~\cite{ChorGKS95}), as well as for the design of information-theoretic proofs systems such as Probabilistically Checkable Proofs (PCPs)~\cite{AroraLMSS92,AroraS92} and Interactive Oracle Proofs (IOPs)~\cite{Ben-SassonCS16,ReingoldRR16} with ZK guarantees (e.g., in~\cite{Ben-SassonCGV16,BCFGRS17,ChiesaFGS18,BBHR19,BCL20,RW24,GurOS25}). 

Many of the early applications of ZK codes relied on ZK properties of \emph{polynomial-based} codes (in which the message is interpreted as a low-degree polynomial, and the corresponding codeword is the evaluation table of the polynomial). More recent work also studied and exploited ZK properties of other codes --- such as concatenated codes~\cite{DecaturGR20}, interleaved codes~\cite{AmesHIV17,BCL20,ChiesaFW26}, 
and tensor codes~\cite{IshaiSVW13,BCL20,RW24} --- with the goal of obtaining improved efficiency such as smaller alphabet size, higher rate, and faster encoding and decoding algorithms. These parameters are tightly connected to the efficiency measures of the resulting applications. Several other works have also studied more generally the parameters achievable by ZK codes~\cite{ChenCGHV07, IshaiSVW13, CDN2015}, 
suggested alternative characterizations of ZK codes~\cite{ChenCGHV07,IshaiSVW13, CDN2015, BCL20}, 
and introduced general frameworks for constructing such codes~\cite{DecaturGR99,DecaturGR20,FeldMSS04, IshaiSVW13}. 

\paragraph{Quantum CSS codes.} 

A quantum CSS code consists of a pair of linear codes $C_X, C_Z \subseteq \FF^n$ so that their dual codes $C_X^\perp, C_Z^\perp$ are orthogonal to each other (that is, $\langle c,c' \rangle=0$ for any $c \in C_X^\perp, c' \in C_Z^\perp$), with the additional guarantee that  
any vector in $C_X \setminus C_Z^\perp$ and any vector in $C_Z \setminus C_X^\perp$ has a large weight.  

Quantum CSS codes were first introduced by Calderbank and Shor~\cite{CS96} and Steane~\cite{Ste96}, and have been extensively studied since then. CSS codes are used to protect quantum computation from errors, and also have applications in quantum complexity theory. 
For example, the recent progress~\cite{ABN23} on the quantum PCP conjecture~\cite{AN02, AALV09} relied on recent breakthroughs on the construction of asymptotically good LDPC quantum CSS codes~\cite{PK22}.
The resemblance of CSS codes to classical codes enables one to draw on the vast literature on classical error-correcting codes for their design, and by now we know of various constructions of such codes.  

\paragraph{Equivalence between zero-knowledge and quantum CSS codes.}

In this note, we show that ZK codes and quantum CSS codes are equivalent. Specifically, we show a transformation from a ZK code into a quantum CSS code (and vice versa) so that the guarantee on the weight of vectors in $C_Z \setminus C_X^{\perp}$ in the CSS code translates into the ZK property of the ZK code, while the guarantee on the 
weight of vectors in $C_X \setminus C_Z^{\perp}$ translates into the decoding property of the ZK code (the latter is also typically required in applications); see Theorem~\ref{thm:equiv} and Corollary~\ref{cor:equiv} for a formal statement of this equivalence. 

While we view this equivalence as interesting in its own right, it is also motivated by the hope that this connection might shed light on our understanding of these codes, or yield new ZK codes or CSS codes by combining the equivalence with existing constructions. We demonstrate this potential by describing one concrete application of the equivalence. Specifically, in Section \ref{s:app}, we translate recent breakthroughs on the construction of asymptotically-good locally testable quantum codes~\cite{DLV24, KP25, WLH25} into explicit asymptotically-good zero-knowledge locally-testable codes.

Finally, we also mention that another connection between quantum computation and zero-knowledge was discovered in~\cite{LacerdaRR19}, who used fault-tolerance in the quantum setting to obtain leakage resilience in the classical computation setting. This gives another indication to the usefulness of exploring the connection between zero-knowledge and quantum computation. 

\section{Preliminaries}\label{s:prelims}
Let $\FF$ be a finite field, and let $u,v\in\FF^n$. We use $\langle u,v\rangle$ to denote the inner product between $u,v$, namely, $\langle u,v\rangle:=\sum_{i=1}^n u_i\cdot v_i$. We also let
 $\Delta(u,v):= \left|\left\{ i\in\left[n\right]: u_i\ne v_i \right\} \right|$,  
 $\wt(u):= \Delta(u,0)$, and for a subset $S \subseteq \FF^n$, we let $\Delta(u,S):=\min_{s\in S}\Delta(u,s)$.

A \textsf{linear (error-correcting) code} is a subspace $C\subseteq\FF^{n}$ over 
$\FF$.  We call $\FF$ and $n$ the \textsf{alphabet} and the
\textsf{block length} of the code, respectively, and the elements of $C$ are called \textsf{codewords}. 
 The \textsf{rate}
of $C $ is the ratio $R:=\frac{\dim(C)}{n}$,  and it measures the amount of redundancy in encoding. 
The \textsf{(Hamming) distance} of $C$ is $\Delta(C):= \min_{c \neq c' \in C} \Delta(c,c')$, which for a linear code equals $\wt(C):=\min_{0\neq c\in  C} \wt(c) $. 
Intuitively, the minimum distance of a code measures the amount of noise tolerance of the code. Specifically, the channel might corrupt some of the entries of a transmitted codeword $c$ during transmission and the receiver might receive a string $w \in \FF^n$. However, if $w$ and $c$ differ on less than $\frac {\Delta(C)} 2$ entries, then $c$ can be recovered uniquely by the decoder by searching for the unique codeword in $C$ that is closest to the   received word $w$. 

A \textsf{generator matrix} for a linear code $C \subseteq \FF^n$ of dimension $k$ is a (full-rank) matrix $G \in \FF^{n \times k}$ so that $\img(G) = C$, and a \textsf{parity-check matrix} for $C$ is a (full-rank) matrix $H \in \FF^{ (n-k) \times n}$ so that $\ker(H)= C$ (note that both $G$ and $H$ are not unique). The \textsf{dual code} of $C$ is the code $C^\perp \subseteq \FF^n$ containing all strings $c' \in \FF^{n}$ satisfying that $\langle c', c\rangle = 0$ for all $c \in C$. 
It follows by definition that $(C^{\perp})^{\perp} = C$, and that $H$ is a parity-check matrix for $C$ if and only if $H^T$ is a generator matrix for $C^{\perp}$.

\subsection{Zero-Knowledge (ZK) Codes}\label{s:zk-code}
Zero-Knowledge (ZK) codes, introduced by Decatur, Goldreich, and Ron~\cite{DecaturGR20}, are codes in which few codeword symbols reveal no information about the message. For this to be possible, we need to associate the code $C$ with a \emph{randomized encoding map}, where $C$ has ZK with respect to this randomized encoding. We will focus on an encoding map that is linear in the message and the randomness used for encoding (such a map was used in most prior works, e.g.,~\cite{DecaturGR20,BCL20,RW24}).

\BD[Randomized Encoding Map]\label{d:rand-enc-c}
Let $C \subseteq \FF^n$ be a linear code of dimension $k$, let $G \in \FF^{n \times k}$ be a generator matrix for $C$, and let $k'<k$ be a parameter.  
The $k'$-\textdef{randomized encoding map for $G$}  is a random map 
$\Enc: \FF^{k'} \to C$, which on input message $\msg \in\FF^{k'}$, samples a uniformly random $r \in \FF^{k-k'}$, and outputs $G\cdot z$, where $z=(\msg,r) \in \FF^k$.
\ED

\noindent The zero-knowledge property of an encoding map is defined as follows. 

\BD[Zero-Knowledge (ZK) Code]\label{d:zk_code}
Let $C \subseteq \FF^n$ be a linear code of dimension $k$, let $G \in \FF^{n \times k}$ be a generator matrix for $C$, and let $k'<k$ and $t<n$ be parameters. We say that the $k'$-randomized encoding map $\Enc: \FF^{k'} \to C$ for $G$ is  $t$-\textsf{Zero-Knowledge ($t$-ZK)} if 
 for any $\cI\subseteq\Set{n}$ of size $t$, and for any pair  of messages $\msg,\msg'\in\FF^{k'}$, we have $\Enc(\msg)|_\cI\equiv\Enc(\msg')|_\cI$. 
\ED 

\paragraph{Error Correction in ZK Codes.}
Applications of zero-knowledge codes also typically require that one can recover the message $\msg \in \FF^{k'}$ from its randomized encoding. Specifically, let $C \subseteq \FF^n$ be a linear code of dimension $k$, let $G \in \FF^{n \times k}$ be a generator matrix for $C$, and let $k'<k$ be a parameter. 
We say that the $k'$-randomized encoding map $\Enc: \FF^{k'} \to C$ for $G$ is \textsf{decodable from $e$ errors} if  there is a (deterministic) algorithm $D$ so that for every message $\msg \in \FF^{k'}$ and for every $y \in \FF^n$ with $\wt(y) \leq e$, it holds that $\Pr[D (\Enc(\msg) + y) =\msg]=1$ 
(see, e.g., \cite[Thm. 1]{DecaturGR20}). Note that the existence of such a (not necessarily efficient) decoding algorithm $D$ is equivalent to the property that for any pair of distinct messages $\msg \neq \msg' \in \FF^{k'}$ and (not necessarily distinct) $r,r' \in \FF^{k-k'}$, it holds that $\Delta( G\cdot z, G \cdot z' )> 2e$, where $z=(\msg,r)$ and $z'=(\msg',r')$. Indeed, if this latter property holds, then given $w:=\Enc(\msg) + y$, the decoder $D$ can find $\msg$ by searching for $z \in \FF^k$ so that $G \cdot z$ is closest to $w$. 

\subsection{Quantum CSS Codes}\label{s:css}

Quantum CSS codes, introduced by Calderbank and Shor~\cite{CS96} and Steane~\cite{Ste96}, are error-correcting codes that allow for quantum error correction, and are defined as follows.

\BD[CSS Code]\label{d:quantum_code}
A \textsf{CSS code} is a pair of linear codes $C_X, C_Z \subseteq \FF^n$ so that the subspaces $C_X^{\perp}$ and $C_Z^{\perp}$ are orthogonal (that is, $\langle c, c'\rangle =0$ for any $c \in C_X^{\perp}$ and $c'\in C_Z^{\perp}$).   The \textsf{rate} of $(C_X,C_Z)$ is $\frac{\dim(C_X) - \dim(C_Z^\perp)} {n} = \frac{\dim(C_Z) - \dim(C_X^\perp)} {n}$.
The \textsf{distance} $d_X$ ($d_Z$, respectively) is defined as the smallest weight of a vector of $C_X$ not in $C_Z^{\perp}$ ($C_Z$ not in $C_X^{\perp}$, respectively). The \textsf{distance} of $(C_X,C_Z)$ is defined as $d=\min \{d_X,d_Z\}$. 
\ED

\section{Zero-knowledge and Quantum CSS Codes Are Equivalent}\label{s:equiv}

We show that ZK and CSS codes (Definitions \ref{d:zk_code} and ~\ref{d:quantum_code}, respectively) are equivalent. 
Specifically, the following theorem shows how to transform ZK codes into CSS codes with similar parameters. 

\begin{thm}[ZK Codes are CSS Codes]\label{thm:equiv}
Let $C \subseteq \FF^n$ be a linear code of dimension $k$, let $G \in \FF^{n \times k}$ be a generator matrix for $C$, let $k'<k$ be a parameter, and let $\Enc: \FF^{k'} \to C$ be the $k'$-randomized encoding map for $G$. Let $C_X = C $, and let $C_Z \subseteq \FF^n$ be the subspace orthogonal to the span of the last $k-k'$ columns of $G$. 
Then the following holds:
\begin{enumerate}
\item $C_X^{\perp}$ and $C_Z^{\perp}$ are orthogonal (so $(C_X,C_Z)$ is a $\mathrm{CSS}$ code).
\item $d_X > 2e$ if and only if  $\Enc$ is decodable from $e$ errors.
\item $d_Z > t$ if and only if $\Enc$ is $t$-ZK. 
\end{enumerate}
\end{thm}

\noindent The inverse transformation, from CSS Codes to ZK codes, follows as a corollary of Theorem~\ref{thm:equiv}:

\begin{corr}[CSS Codes are ZK Codes]\label{cor:equiv}
Let $C_X, C_Z \subseteq \FF^n$ be linear codes 
so that $C_X^{\perp}$ and $C_Z^{\perp}$ are orthogonal (i.e., $(C_X,C_Z)$ is a CSS code). Let $k = \dim(C_X)$ and $k'= k- \dim(C_Z^\perp)$.
Let $C=C_X$, let $G \in \FF^{n \times k}$ be a generator matrix for $C$ whose last $k-k'$ columns form a basis for $C_Z^\perp$,\footnote{Such a $G$ exists since $C_Z^\perp \subseteq (C_X^\perp)^\perp=C_X$.}, 
and let $\Enc: \FF^{k'} \to C$ be the $k'$-randomized encoding map for $G$.
Then the following holds:
\begin{enumerate}
\item $d_X > 2e$ if and only if  $\Enc$ is decodable from $e$ errors.
\item $d_Z > t$ if and only if $\Enc$ is $t$-ZK. 
\end{enumerate}
\end{corr}

\noindent The above Corollary~\ref{cor:equiv} follows from Theorem~\ref{thm:equiv} by noting that $C_Z^\perp$ is exactly the span of the last $k-k'$ columns of $G$. We therefore turn our attention to proving Theorem~\ref{thm:equiv}. The proof relies on the following equivalent definition of a ZK code. 

\begin{lem}[ZK Codes, Equivalent Formulation]\label{lem:equiv}
Suppose that $C \subseteq \FF^n$ is a linear code of dimension $k$, let $G \in \FF^{n \times k}$ be a generator matrix for $C$, and let $k'<k$ and $t<n$ be positive integers. Then the $k'$-randomized encoding map $\Enc: \FF^{k'} \to C$ 
for $G$ is $t$-ZK if and only if any linear combination of any $t$ rows in $G$ does not result in a non-zero $w \in \FF^k$ such that $w|_{[k]\setminus [k']}=0$.
\end{lem}

A similar characterization of ZK codes as in the above Lemma~\ref{lem:equiv} was given in~\cite[Lemma 5.1]{BCL20}. Here we provide an alternate self-contained proof. We also note that~\cite[Claim 6.1]{IshaiSVW13} show that the {\em stronger} condition that ``any linear combination of $t$ rows in $G$ does not result in ({\em not} necessarily zero) $w \in \FF^k$ such that $w|_{[k]\setminus [k']}=0$" implies that $\Enc$ is $t$-ZK. Finally, we note that~\cite[Lemma 5.2]{BCL20} showed that a similar condition to the latter is equivalent to the stronger property that $\Enc$ is \emph{uniform} $t$-ZK. ($\Enc$ is \emph{uniform $t$-ZK} if for any $\msg \in \FF^{k'}$ and for any $\cI \subseteq [n]$ of size $t$, $\Enc(\msg)|_\cI$ is the uniform distribution over $\FF^t$.)

\paragraph{Proof of Lemma~\ref{lem:equiv}:}
For $\cI \subseteq [n]$, let $G|_\cI $ denote the restriction of $G$ to the rows in $\cI$.
It suffices to show that  for any subset $\cI\subseteq\Set{n}$, $\Enc(\msg)|_\cI \equiv \Enc(\msg')|_\cI$ for any pair  of  messages $\msg,\msg'\in\FF^{k'}$ if and only if any linear combination of the rows of $G|_\cI$  does not result in a non-zero $w \in \FF^k$ such that $w|_{[k]\setminus [k']}=0$. 
We prove the lemma in two steps. First, we show (in Claim~\ref{clm:equiv_1}) that the former requirement ``$\Enc(\msg)|_\cI \equiv \Enc(\msg')|_\cI$ for any $\msg,\msg' \in \FF^{k'}$'' is equivalent to requiring that ``$0 \in \Supp(\Enc(\msg)|_\cI)$ for any $\msg \in \FF^{k'}$". Then, we show (in Claim~\ref{clm:equiv_2}) that the requirement ``$0 \in \Supp(\Enc(\msg)|_\cI)$ for any $\msg \in \FF^{k'}$'' is equivalent to the latter requirement ``any linear combination of rows of $G|_\cI$ does not result in $0\neq w \in \FF^k$ with $w|_{[k]\setminus [k']}=0$''. This will conclude the proof of the lemma.

\begin{clm}\label{clm:equiv_1}
Let $\cI\subseteq\Set{n}$. Then $\Enc(\msg)|_\cI \equiv \Enc(\msg')|_\cI$ 
for any pair of  messages $\msg,\msg'\in\FF^{k'}$  if and only if $0 \in \Supp(\Enc(\msg)|_\cI)$ for any $\msg \in \FF^{k'}$.
\end{clm}

\begin{proof}
For the 'only if' part, note that if $0 \notin \Supp(\Enc(\msg)|_\cI)$ for some $0 \neq \msg \in \FF^{k'}$, then since $0 \in \Supp(\Enc(0)|_\cI)$, then we clearly have that $\Enc(\msg)|_\cI \not\equiv \Enc(0)|_\cI$. 

\sloppy
For the 'if' part, assume that $0 \in \Supp(\Enc(\msg)|_\cI)$ for any $\msg \in \FF^{k'}$. Then in this case, for any $\msg,\msg'\in\FF^{k'}$, we have that $0 \in \Supp(\Enc(\msg'-\msg)|_\cI)$, and so by linearity, 
 $$\Supp(\Enc(\msg')|_\cI) \supseteq \Supp(\Enc(\msg)|_\cI)+ \Supp(\Enc(\msg'-\msg)|_\cI) \supseteq  \Supp(\Enc(\msg)|_\cI).$$ 
 Indeed, the right containment uses the fact that $0 \in \Supp(\Enc(\msg'-\msg)|_\cI)$. The left containment follows from linearity, because if $u:=G|_\cI\cdot (\msg, r) \in \Supp(\Enc(\msg)|_\cI)$ and $u': = G|_\cI\cdot (\msg'-\msg, r') \in \Supp(\Enc(\msg'-\msg)|_\cI)$, then $u + u' = G|_\cI \cdot (\msg', r+r') \in \Supp(\Enc(\msg')|_\cI)$. 
 In summary, $\Supp(\Enc(\msg')|_\cI) \supseteq  \Supp(\Enc(\msg)|_\cI)$ for any pair of messages $\msg,\msg'\in\FF^{k'}$, so we conclude that  $\Supp(\Enc(\msg)|_\cI) =  \Supp(\Enc(\msg')|_\cI) $ for any pair of messages $\msg,\msg'\in\FF^{k'}$. 
 Finally, observe that  by properties of linear algebra, for any $\msg \in\FF^{k'}$, the number of $r \in \FF^{k-k'}$ which satisfy the system of linear equations $G|_\cI \cdot (\msg, r)=v$ is the same for any $v \in \Supp(\Enc(\msg)|_\cI)$. Consequently, 
 the fact that  $\Supp(\Enc(\msg)|_\cI) =  \Supp(\Enc(\msg')|_\cI) $  implies that 
 $\Enc(\msg)|_\cI \equiv \Enc(\msg')|_\cI$, and we conclude that $\Enc(\msg)|_\cI \equiv \Enc(\msg')|_\cI$ 
for any pair of  messages $\msg,\msg'\in\FF^{k'}$. 
\end{proof}

\begin{clm}\label{clm:equiv_2}
Let $\cI \subseteq [n]$. Then $0 \in \Supp(\Enc(\msg)|_\cI)$ for any $\msg \in \FF^{k'}$
if and only if any linear combination of the rows of $G|_\cI$  does not result in a non-zero $w \in \FF^k$ such that $w|_{[k]\setminus [k']}=0$.
\end{clm}

\begin{proof}
For the 'only if' part, suppose that there exists a linear combination $u \in \FF^{|\cI|}$
of the rows of $G|_\cI$ which results in a non-zero $w \in \FF^k$ such that $w|_{[k]\setminus [k']}=0$, and let $j \in [k']$ be an entry so that $w_j \neq 0$. Let $\msg \in \FF^{k'}$ be the $j$-th unit vector. We shall show that there does not exist an $r \in \FF^{k-k'}$ so that  
 $G|_\cI \cdot (\msg,r)= 0$, and consequently $0 \notin \Supp(\Enc(\msg)|_\cI)$. To see the latter, suppose on the contrary that there exists an $r \in \FF^{k-k'}$ so that  
 $G|_\cI \cdot (\msg,r)= 0$. Then we have that 
 $$0= \langle u , G|_\cI \cdot (\msg,r) \rangle  = \langle (u^T \cdot G|_\cI)^T , (\msg,r) \rangle = \langle w,  (\msg,r)\rangle = \langle w|_{[k']}, \msg \rangle + \langle w|_{[k]\setminus [k']}, r\rangle =w_j +0 = w_j \neq 0, $$
 which is a contradiction. 

For the 'if' part, suppose that there exists an $\msg \in \FF^{k'}$ so that $0 \notin \Supp(\Enc(\msg)|_\cI)$. Then the linear system $G|_\cI \cdot (\msg, r)=0$ does not have a solution $r \in \FF^{k-k'}$. Let $A$ be the matrix which consists of the first $k'$ columns of $G|_\cI$, and let $B$ be the matrix which consists of the last $k-k'$ columns of $G|_\cI$. Then by properties of linear algebra, we have that there exists a linear combination $u \in \FF^{|\cI|}$ so that $\langle u, A \cdot \msg \rangle
= \langle (u^T \cdot A)^T, \msg \rangle
\neq 0$ but $u^T \cdot B =0$. 
(Indeed, the system $G|_\cI \cdot (\msg, r)=0$ is equivalent to the system $B \cdot r = -A \cdot \msg$, where $A,B,\msg$ are fixed.)
But this implies in turn that 
$w:=u^T \cdot G|_\cI = ( u^T \cdot A, u^T \cdot B)$ is a linear combination of the rows of $G|_\cI$ which satisfies that $w|_{[k']}  = u^T \cdot A \neq 0$, but $w|_{[k]\setminus [k']} = u^T \cdot B =0$. 
\end{proof} 
\noindent This concludes the proof of Lemma~\ref{lem:equiv}. \qed

\vspace{2mm}
\noindent We now turn to the proof of Theorem~\ref{thm:equiv}, based on the above Lemma~\ref{lem:equiv}. 

\paragraph{Proof of Thm~\ref{thm:equiv}:}

We prove each of the items separately.

\medskip \noindent \textbf{Item (1):} Follows since $C_X^{\perp} = C^{\perp}$, and $C_Z^{\perp}$ is the span of the last $k-k'$ columns of $G$, and so $C_Z^{\perp} \subseteq C$. 

\medskip \noindent \textbf{Item (2):} Recall that $\Enc$ is decodable from $e$ errors if and only if for any pair of distinct messages $\msg \neq \msg' \in \FF^{k'}$ and (not necessarily distinct) $r,r' \in \FF^{k-k'}$, it holds that 
$\Delta(G\cdot (\msg,r), G \cdot (\msg',r')) >2e$. Further, by linearity, this latter property is equivalent to the property that for any non-zero $\msg  \in \FF^{k'}$ and (possibly zero) $r \in \FF^{k-k'}$ it holds that $\wt(G \cdot (\msg,r)) >2e$.
Thus, it suffices to show that 
$d_X =\min_{0 \neq \msg  \in \FF^{k'}, r \in \FF^{k-k'}} \wt(G\cdot (\msg,r))$.  

\noindent But the above follows since 
\begin{eqnarray*}
C_X \setminus C_Z^\perp & =&  \{c \in C_X \mid c \notin C_Z^\perp\} \\
& =&  \{c \in C \mid c \; \text{is not in the span of the last $k-k'$ columns of $G$}\; \} \\
& = & \{G\cdot (\msg,r) \mid 0 \neq \msg \in \FF^{k'}, r \in \FF^{k-k'}\},
\end{eqnarray*}
and so $d_X = \min_{c \in C_X \setminus C_Z^\perp} \wt(c) = \min_{0 \neq \msg  \in \FF^{k'}, r \in \FF^{k-k'}} \wt(G\cdot (\msg,r))$.

\medskip \noindent \textbf{Item (3):} By Lemma~\ref{lem:equiv}, $\Enc$ is $t$-ZK if and only if any linear combination of any $t$ rows in $G$ does not result in a non-zero $w \in \FF^k$ such that $w|_{[k]\setminus [k']}=0$. We shall show that the latter condition is equivalent to the condition that $d_Z >t$. To see this, note that a linear combination of $t$ rows in $G$ resulting in $w \in \FF^k$, corresponds to a vector $u \in \FF^n$ of weight at most $t$ so that $u^T \cdot G=w$. Furthermore, the condition that $w \neq 0$ is equivalent to the condition that $u \notin C_X^\perp = C^\perp$, while the condition that $w|_{[k]\setminus [k']}=0$ is equivalent to the condition that $u \in C_Z$. Thus, the condition that  any linear combination of any $t$ rows in $G$ does not result in a non-zero $w \in \FF^k$ such that $w|_{[k]\setminus [k']}=0$ is equivalent to the condition that there do not exist $u\in C_Z \setminus C_X^{\perp}$ of weight at most $t$,  which is equivalent to the condition that $d_Z>t$. 
\qed 

\section{Application: Explicit Asymptotically-Good Zero-knowledge Locally-Testable Codes}\label{s:app}
The equivalence between ZK codes and quantum CSS codes of Section~\ref{s:equiv} is particularly attractive from an efficiency prospective, since it preserves the properties of the code.
We now describe an immediate application of this equivalence, using recent constructions of asymptotically-good locally testable quantum codes to obtain explicit asymptotically-good ZK codes that are also {\em locally testable} with a few queries. We first formally define locally-testable codes (LTCs).

\BD[Locally-Testable Code (LTC)]\label{def:ltc} A code $C\subseteq \FF^n$ is a \textdef{$q$-query Locally Testable Code ($q$-LTC)} if there exists a randomized oracle 
algorithm $\test$ 
which  receives oracle access to a string $w\in\FF^{n}$, makes $q$ queries to $w$, and outputs either 'accept' or 'reject', so that the following conditions holds:
\begin{itemize}
\item \textbf{Completeness:} If $w\in C$ then $\test$ accepts with probability $1$.
\item \textbf{Soundness:} If $w\notin C$ then $\test$ rejects with probability at least 
$\frac 1 4 \cdot \frac{\Delta(w,C)} {n}$. 
\end{itemize}
\ED

A ZK code that is also an LTC, with a ZK threshold that is significantly larger than the query complexity of the local tester, is called a {\em ZK-LTC}.
Such codes lie at the heart of ZK-PCP and ZK-IOP constructions (and are also used in other cryptographic contexts such as verifiable secret sharing). Ishai et al.~\cite{IshaiSVW13} gave a generic method of (probabilistically) transforming any linear code into a ZK code (Their probabilistic transformation outputs a generator matrix for the ZK code, with negligible probability of error.) They then use this transformation to obtain a \emph{probabilistic}  construction of asymptotically good ZK-LTCs. 
Combining new quantum LTC constructions~\cite{DLV24,KP25,WLH25,BafnaLN26,CHLT, CohenLRTT26,GayJ26}, 
and the equivalence between quantum CSS codes and ZK codes, yields new ZK-LTCs in several parameter regimes. In particular, it gives an {\em explicit} construction of asymptotically-good ZK-LTCs surpassing the efficiency of the probabilistic construction of~\cite{IshaiSVW13}. This is formalized in Corollary~\ref{cor:zk-ltc} below. We first cite the relevant quantum LTCs.
%
Several recent works have constructed asymptotically-good quantum LTCs with constant locality:\footnote{We note that in some of these papers, the result is stated as a conditional result. We refer the interested reader to the papers themselves for these conditions and restrictions.} 
\BT[Asymptotically-good quantum LTCs with constant locality~\cite{BafnaLN26,CHLT, CohenLRTT26,GayJ26}]\label{t:css-codes-constant}
There exists an explicit infinite family of CSS codes $C_X,C_Z \subseteq \FF_2^n$, where $(C_X,C_Z)$ has constant rate, distance $\Omega(n)$, and both  $C_X$ and $C_Z$ are locally testable with  $O(1)$ queries.
\ET

Combining the above Theorem~\ref{t:css-codes-constant} with Corollary~\ref{cor:equiv} gives explicit asymptotically-good ZK-LTCs with linear ZK threshold that are locally testable with a constant number of queries. To the best of our knowledge, this is the first instance of an explicit family of asymptotically-good ZK-LTCs in which the ZK threshold is larger than the tester's query complexity. 
\BCR\label{cor:zk-ltc}
There exists an explicit infinite family of codes $\mathcal{C}=\left(C_n\right)_n$, where  $C_n \subseteq \FF_2^n$ is a linear code that is locally testable with $O(1)$ queries. Furthermore, there exist an explicit generator matrix $G$ for $C_n$ and $k'= \Theta(n)$, 
so that the $k'$-randomized encoding map $\Enc: \FF_2^{k'} \to C_n$ for $G$ is $\Omega(n)$-ZK, and is
decodable from $\Omega(n)$ errors. 
\ECR

\begin{remarkk}
We note that asymptotically-good ZK-LTCs with slightly worse parameters compared to Corollary~\ref{cor:zk-ltc} follow, using our equivalence of Section~\ref{s:equiv}, already from quantum CSS codes in the earlier works of~\cite{DLV24,KP25,WLH25}. Specifically, Dinur, Lin, and Vidick~\cite{DLV24} construct explicit binary quantum LTCs with constant rate, $\Omega(1/\log^3 n)$ relative distance, and a local tester that makes $O(1)$ queries, but 
the rejection probability is $\Omega(1/\log^3 n)$. This gives explicit binary ZK-LTCs with 
$\Omega(n/log^3n)$-ZK and the same parameters as the quantum code. Their codes can also be made to have 
constant soundness at the price of increasing the tester's query complexity to $O(\log^3 n)$, which gives explicit binary ZK-LTCs with 
$\Omega(n/log^3n)$-ZK, constant rate, $\Omega(1/\log^3 n)$ relative distance, and a local tester that makes $O(\log^3 n)$ queries. 
\end{remarkk}

\begin{remarkk}
We note that quantum LTCs satisfy the stronger requirement that \emph{both} $C_X$ and $C_Z$ are locally testable, while the application for ZK-LTCs only requires that $C=C_X$ is locally testable. 
\end{remarkk}

\paragraph{Acknowledgement.} We thank Louis Golowich, Thomas Vidick, and  Gilles Z{\'{e}}mor for discussions on the topic. The discussion leading to this note was initiated at the ``Error-Correcting Codes: Theory and Practice Reunion'' held at the Simons Institute for the Theory of Computing on April 2025. Research supported in part by a grant from the UC Noyce Initiative to the Simons Institute for the Theory of Computing. 

Noga Ron-Zewi was partially supported by the European Union (ERC, ECCC, 101076663). Views and opinions expressed are however those of the author(s) only and do not necessarily reflect those of the European Union or the European Research Council. Neither the European Union nor the granting authority can be held responsible for them. The second author was supported by the ISRAEL SCIENCE FOUNDATION (grant No. 434/24).

\bibliography{bibliography}
\bibliographystyle{alpha}

\appendix

\end{document}